\documentclass{comjnlarxiv}

\usepackage{flushend}

\usepackage{url}
\usepackage{setspace}
\usepackage{listings}
\usepackage{multirow} 
\usepackage{ulem}
\usepackage{subfig}
\usepackage{amssymb}
\usepackage{amsmath}
\usepackage{algorithm}
\usepackage{algorithmic}
\usepackage{verbatim}
\newcommand{\citep}[1]{\cite{#1}}
\begin{document}

\title[Strongly universal string hashing is fast]{Strongly universal string hashing is fast}
\shortauthors{D. Lemire and O. Kaser}

\author{Daniel Lemire} \email{lemire@gmail.com} \affiliation{LICEF Research Center, TELUQ, Universit\'e du Qu\'ebec, Canada}\address{} 
\author{Owen Kaser} \affiliation{Department of CSAS, University of New Brunswick, Canada}

%%%%%%% computer journal

\begin{abstract}
We present fast strongly universal string hashing families: 
they can process data at a rate of  0.2~CPU cycle per byte. 
Maybe surprisingly, we find that these families---though they require a large
buffer of random numbers---are often faster than popular hash functions with weaker
theoretical guarantees.
Moreover, conventional wisdom is that hash functions with fewer multiplications are faster.
Yet we find that they may fail to be faster
due to operation pipelining. 
 We present experimental results on 
several processors including %low-powered [elsewhere we used low-power]
low-power 
processors.
Our tests include
hash functions designed for processors with the 
Carry-Less Multiplication (CLMUL) instruction set.
 We also
prove, using accessible proofs, the strong universality of our families.
\end{abstract}

\maketitle

\section{Introduction}
For 32-bit integers, random hashing with good theoretical guarantees
can be just as fast as popular alternatives~\citep{338597}. 
In turn, these guarantees ensure the reliability of various algorithms
and data structures: set intersection~\citep{Ding:2011:FSI:1938545.1938550},
 frequent-item mining~\citep{Cormode:2010:MFF:1731351.1731356},
count estimation~\citep{Gibbons2001,BarYossef2002},
and  
hash tables~\citep{pagh2004cuckoo,Dietzfelbinger:2009:RUC:1496770.1496857,Aumuller:2012:EEH:2404160.2404171}. 
We want to show
that we can also get good theoretical guarantees over 
larger objects (such as strings) without sacrificing speed.  
For example, we consider
variable-length strings made of 32-bit characters: all data structures can
be represented as such strings, up to some padding.

We restrict 
our attention to hash functions mapping strings to $L$-bit integers, that is, integers in $[0,2^L)$ for some positive integer $L$.
In random hashing, we select a hash function at random from a family~\citep{carter1979uch,wegman1981new}. 
The hash function can be chosen whenever the software is initialized. While random hashing is not yet commonplace, it can have significant security benefits~\cite{Crosby:2003:DSV:1251353.1251356} in a hash table: without randomness, an attacker can more easily exploit the fact that adding $n$~keys hashing to the same value 
typically
takes quadratic time ($\Theta(n^2)$). 
For this reason, random hashing was adopted in the Ruby language as of version~1.9~\cite{ocert2011003} and in the Perl language as of version~5.8.1. 

A family of hash functions is $k$-wise independent (or $k$-independent) if 
the hash values of any $k$~distinct elements are independent. For example,
a family is pairwise independent---or strongly universal---if given any two distinct elements $s$ and $s'$, 
their hash values $h(s)$ and $h(s')$ are independent: \begin{eqnarray*}P(h(s)=y | h(s')=y') = P(h(s)=y)\end{eqnarray*} for any two
hash values $y,y'$. (Some authors prefer the terms 2-independent or 2-universal to describe strongly universal hash families.) 
When a hashing family is not strongly universal,
it can still be universal if 
the probability of a collision is 
no larger than if it were strongly universal: $P(h(s)=h(s'))\leq 1/2^L$ when $2^L$ is the number of hash values. 
If the collision probability is merely bounded by some $\epsilon$ larger than $1/2^L$ but smaller than $1$ ($P(h(s)=h(s'))\leq \epsilon<1$), we have an almost universal family. 
However, strong universality might be more desirable than universality or almost universality:
\begin{itemize}
\item 
We say that a family is uniform if all hash values are equiprobable ($P(h(s)=y)=1/2^L$ for all $y$ and $s$): strongly universal families are uniform, but  universal or almost universal families may fail to be uniform. 
To see that universality fails to imply uniformity, consider
the family made of the two functions 
over 1-bit integers (0,1): the identity and a function mapping all
values to zero. The probability of a collision between two distinct values is exactly $1/2$ which ensures universality even though we do not have uniformity since $P(h(0)=0)=1$. 
 \item 
Moreover, if we have strong universality over  $L$~bits, then we also have it over any subset of bits. The corresponding result may fail for universal and almost universal families: we might have universality over $L$~bits, but fail to have almost universality over some subset of bits. 
Consider the non-uniform but universal family $\{ h(x)= x\}$ over $L$-bit integers:  if we keep only the least significant $L'$~bits ($0<L'<L$), universality is lost
since  $h(0)\bmod {2^{L'}} = h(2^{L'}) \bmod {2^{L'}}$.
\end{itemize}

There is no need to use slow operations such as 
modulo operations,
 divisions or operations in finite fields to have strong universality. 
In fact, for short strings having few distinct characters,   Zobrist hashing
requires nothing more than table look-ups and bitwise exclusive-or operations, and it is more than strongly
universal (3-wise independent)~\citep{zobrist1970,zobrist1990new}. 
Unfortunately, 
it becomes prohibitive for long strings as it requires the storage of $n c$~random numbers where
$n$ is the maximal length of a string and $c$ is the number of distinct characters.

A more practical approach to strong universality is Multilinear hashing (\S~\ref{sec:multilinear}). 
Unfortunately, it normally requires that the computations be executed in
a finite field. Some processors have instructions for finite fields (\S~\ref{sec:carryless}) or they can be emulated with a software library (\S~\ref{sec:experimentsonfinitefields}).
However, if we are willing to double the number of random bits,
we can implement it using regular integer arithmetic. 
Indeed, using an idea from Dietzfelbinger~\cite{dietzfelbinger1996universal}, 
we 
implement it using only one multiplication and one addition per character (\S~\ref{sec:technical}).
We further attempt to speed it up 
by reducing the number of multiplications by half. 
We believe that these families are the fastest strongly universal hashing
families on current computers.  
We evaluate these hash families experimentally~(\S~\ref{sec:experiments}):
\begin{itemize}
\item 
Using fewer multiplications has often  improved performance, 
especially on low-power processors~\citep{etzel1999square}.
Yet 
trading away the number of multiplications fails to improve (and may even degrade) performance 
on several processors according to our experiments---which include low-power processors. However,  reducing the number of multiplications is beneficial on some processors (e.g., AMD V120).
\item We also find that strongly universal hashing may be computationally inexpensive compared to common hashing functions, as long as
we ignore the overhead of generating long strings of random numbers.
In effect---if memory is abundant compared to the length of the strings---the strongly universal Multilinear family is 
faster than many of the commonly
used alternatives.

\item  
We consider hash functions  designed for hardware supported carry-less multiplications (\S~\ref{sec:carryless}). This support
should drastically improve the speed of some operations over binary finite fields ($GF(2^L)$). Unfortunately, we find that the carry-less hash functions fail to be competitive (\S~\ref{sec:carryless-experiments}). 
\end{itemize}

\section{The Multilinear family}
\label{sec:multilinear}

The Multilinear hash family is one of the simplest strongly universal families~\citep{carter1979uch}. 
It takes the form of a scalar product between random values (sometimes called keys) and string components, where operations are over a finite field: \begin{eqnarray*}h(s)=m_1+\sum_{i=1}^n m_{i+1} s_i.\end{eqnarray*}
The hash function $h$ is defined by the  randomly generated values $m_1, m_2, \dots$  It is strongly universal over fixed-length strings. We can also apply it to variable-length strings
as long as we forbid strings ending with zero. To ensure that strings never end with zero, we can append a character value of one to all variable-length strings.

An apparent limitation of this approach
is that strings cannot exceed the number of random values. In effect, to hash 32-bit strings of length $n$, we need to generate and store $32 ( n + 1)$~random bits using a finite field of cardinality $2^{32}$. 
However, Stinson~\cite{188765} showed that strong universality 
requires at least $1+a(b-1)$~hash functions where $a$ is the number of strings and $b$ is the number of hash values. Thus, if we have 32-bit strings mapped to 32-bit hash values, we need at least $\approx 2^{32 ( n + 1)}$~hash functions: Multilinear is almost optimal.

 Hence, the requirement to store  many random numbers cannot be waived without sacrificing strong universality.
Note that Stinson's bound is not affected by manipulations such as treating a length~$n$ string
of $W$-bit words as a length~$n/2$ string of $2W$-bit words.

If multiplications are expensive and we have long strings, we can
attempt to improve speed by reducing the number of multiplications by half~\citep{motzkin1955evaluation,703969}: 
\begin{eqnarray}\label{mhm-gf}h(s)=m_1+\sum_{i=1}^{ n/2 } (m_{2i} + s_{2i-1})(m_{2i+1} + s_{2i}).\end{eqnarray} 
Indeed, this new form follows from the fact that we can rewrite the scalar product $m_{2i} s_{2i} + m_{2i+1}  s_{2i-1}$ as a single multiplication $(m_{2i} + s_{2i-1})(m_{2i+1} + s_{2i})$ minus two terms, one that does not depend on the string ($m_{2i} m_{2i+1}$) and one that does not depend on the random keys ( $s_{2i-1} s_{2i}$).
While this new form assumes that the number of characters in the string is even, we can simply
pad the odd-length strings with an extra character with value zero. With variable-length strings, the padding to even length must follow the addition of a character value of one. 

Could we reduce the number of multiplications further? Not in general: the computation
of a scalar product
between two vectors of length $n$ requires at least $\lceil n/2 \rceil$~multiplications~\cite[Corollary~4]{winograd1970number}.
However, we could try to avoid generic multiplications altogether and replace
them by squares~\citep{etzel1999square}: 
\begin{eqnarray*}h(s)=m_1+\sum_{i=1}^{n} (m_{i+1} + s_{i})^2.\end{eqnarray*}
Indeed, squares can be sometimes be computed faster. 
Unfortunately, this approach fails in binary finite fields  ($GF(2^L)$) because
\begin{eqnarray*}(m_{i+1} + s_{i})^2 & = & m_{i+1}^2 + m_{i+1} s_{i} + m_{i+1} s_{i}  + s_{i}^2 \\ & = & m_{i+1}^2 +  s_{i}^2\end{eqnarray*} since every element is its own additive inverse. Thus, we   get 
\begin{eqnarray*}h(s)=m_1+\sum_{i=1}^{n} m_{i+1}^2 +  \sum_{i=1}^{n} s_{i}^2\end{eqnarray*}
which is  a poor hash function (e.g., 
$h(\texttt{ab})=h(\texttt{ba})$).

There are fast algorithms to compute  multiplications~\citep{Brent:2008:FMG:1789715.1789728,lemi:one-pass-journal,Greenangalois} in binary finite fields. Yet 
these operations remain much slower than a
native operation (e.g., a regular 32-bit integer multiplication). 
However, some recent
processors have support for finite fields. In such cases,
the penalty could be small for using finite fields, as opposed to regular integer
arithmetic (see 
\S~\ref{sec:carryless} and \S~\ref{sec:carryless-experiments}). (Though they are outside our scope, there are also fast techniques for computing hash functions over a finite field having prime cardinality~\cite{krovetz2001fast}.)

\section{Making Multilinear strongly universal in the ring $\mathbb{Z}/2^{K}\mathbb{Z}$}
\label{sec:technical}

On processors without support for binary finite fields, we can trade memory for speed to essentially get the same properties as finite fields on \emph{some} of the bits using fast integer arithmetic. 
For example,
Dietzfelbinger~\cite{dietzfelbinger1996universal} showed that the family of hash functions of the form 
\begin{eqnarray*}h_{A,B}(x) =  \left (Ax +B \mod{2^{K}}\right )  \div 2^{L-1}\end{eqnarray*}  where the integers $A,B \in [0,2^K)$ and $x \in [0,2^L)$
is strongly universal for $K > L-1$. 
(%For fewer parentheses,  owen adds words so math isn't split across line
To reduce the number of parentheses used,
we adopt the convention that
$Ax +B \mod{2^{K}} \equiv (Ax +B) \bmod{2^{K}}$. The symbol  $\div$ denotes integer division: $x \div y = \lfloor x/y \rfloor$ for positive integers.)
 We generalize Dietzfelbinger hashing from the linear to the multilinear case. 

The main difference between
a finite field and common integer arithmetic (in the integer ring $\mathbb{Z}/2^K \mathbb{Z}$)
 is that elements of fields have inverses: given the equation $ax=b$, there is a unique solution $x= a^{-1}b$ when $a\neq 0$. However, the same is ``almost'' true 
in integer rings used for computer arithmetic as long as the variable $a$ is small.
For example, when $a=1$, we can solve for $ax=b$ exactly ($x=b$). When $a=2$, then there are at most two solutions to the equation $ax=b$. We build on these observations to derive a stronger result.

We let $\tau=\textrm{trailing}(a)$  be the number of trailing zeros of the
integer $a$ in binary notation.
For example, we have that
$\textrm{trailing}(2^j)=j$. 

\begin{proposition} 
\label{prop:almostinvertible}
Given integers $K,L$ satisfying 
$K\geq L-1\geq 0$, consider the equation 
  \begin{eqnarray*}\left ( ax+c \mod{2^{K}} \right ) \div 2^{L-1} = b
  \end{eqnarray*} where $a$ is an integer in $[1,2^L)$, $b$ is an integer in $[0,2^{K-L+1})$ and 
  $c$ an integer in $[0,2^K)$. Given $a$, $b$ and $c$, there are exactly $2^{L-1}$~integers $x$ in $[0,2^K)$ satisfying the equation.
\end{proposition}
\begin{proof}
Let $\tau=\textrm{trailing}(a)$. We have $\tau \leq L-1$ since $a\in [1,2^L)$.
Because $a$ is non-zero, 
we have that $a'= a \div 2^{\tau}$ is odd
 and  
 $a= 2^{\tau} a'$.

We have\begin{eqnarray*}
\begin{split}
  \big ( (ax+c) & \mod{2^{K}} \big  )  \div 2^{L-1}\\
& =  \left   ( (2^{\tau} a' x+c) \bmod{2^{K}} \right ) \div 2^{L-1}\\
&=  \left ( 2^{\tau} [ a' x+(c\div 2^{\tau})] \right . \\
& \quad+ \left (c\bmod {2^{\tau}}) \mod{2^{K}} \right) \div 2^{L-1}.
\end{split}
  \end{eqnarray*}

We show that the term  $(c\bmod {2^{\tau}})$ can be removed.
Indeed, consider that
%try to fit math by juggling words (be verbose)
the $\tau$~least significant
bits of 
$2^{\tau} [ a' x+(c\div 2^{\tau})] + (c \bmod 2^{\tau})$
are
those of $c \mod {2^{\tau}}$, and % whereas --- change to fit
the more significant bits are those of  $2^{\tau} [ a' x+(c\div 2^{\tau})]$.
The final division by 
 ${2^{L-1}}$ 
will dismiss the $L-1$~least significant bits, and $\tau \leq L-1$, so that the term
$(c\bmod {2^{\tau}})$ can be ignored.

Hence, we have
\begin{eqnarray*}
\begin{split}
  \big ( ax+c &\mod{2^{K}} \big ) \div 2^{L-1}  \\
 & =   ( 2^{\tau} [ a' x+(c\div 2^{\tau})]  \mod{2^{K}}) \div 2^{L-1}\\
& =  \left(2^{\tau}\left [   a' x+(c\div 2^{\tau})  \mod{2^{K-\tau}}\right ]\right ) \div 2^{L-1}  \\
&=  \left (   a' x+(c\div 2^{\tau})  \mod{2^{K-\tau}}\right ) \div 2^{L-1-\tau}\\
&=  \left (   a' (x \bmod{2^{K-\tau}}) \right . \\
&\quad+ \left (c\div 2^{\tau})  \mod{2^{K-\tau}}\right ) \div 2^{L-1-\tau}. 
\end{split}
  \end{eqnarray*}
  Setting $x'=x \bmod{2^{K-\tau}}$ and
  $c'=c\div 2^{\tau}$, we finally have
\begin{eqnarray*} \begin{split} \big ( ax+c &\mod{2^{K}}\big ) \div 2^{L-1}  \\ &=   \left (   a' x'+c'\mod{2^{K-\tau}} \right ) \div 2^{L-1-\tau}. \end{split}\end{eqnarray*} 
Let $z$ be an integer  such that $z \div 2^{L-1-\tau}=b$.
Consider $a' x' +c' \mod{2^{K-\tau}}=z$. We can rewrite it as
$a' x'  \bmod{2^{K-\tau}}=z- c' \bmod{2^{K-\tau}}$. 
Because $a'$ is odd, 
$a'$ and $2^{K-\tau}$ are coprime (their greatest common divisor is 1).
Hence, there is a unique integer
$x'\in [0,2^{K-\tau})$ such 
that $a' x'  \bmod{2^{K-\tau}}=z - c' \bmod{2^{K-\tau}}$~\cite[Cor. 31.25] {Cormen:2009:IAT:1614191}. 
 
Given $b$, there are $ 2^{L-1-\tau}$~integers $z$ such that
$z \div 2^{L-1-\tau}=b$. 
Given $x'$, there are $2^{\tau}$~integers $x$ in $[0,2^K)$ such that  $x'=x \bmod{2^{K-\tau}}$. 
It follows that there are $ 2^{L-1-\tau} \times 2^{\tau} = 2^{L-1}$~integers $x$ in $[0,2^K)$ such that  
$\left (   a' x'+c'\mod{2^{K-\tau}} \right ) \div 2^{L-1-\tau} = b$ holds.
\end{proof}

\begin{example}
Consider%
, for instance, %verbosity to adjust equation break
the equation $(6x+10 \bmod 64) \div 4 = 5$.
 By Proposition~\ref{prop:almostinvertible}, there must be exactly 4~solutions
to this equation (setting $K=6, L=3$). We can find
them using the proof of the lemma.
The integer $6$ has 1~trailing zero in binary notation ($110$)
so that $\tau=1$. We can write $6=2\times 3$ so that $a'=3$.
Similarly, $c'=10 \div 2=5$. Hence we must consider the
equation $3 x' + 5 \bmod{2^5} = z$ for values of $z$ such that
$z \div  2 = 5$. There are two such values: $z=10$ and $z=11$. We
have that $$\begin{array}{l}
3 x' + 5 \bmod{32} = 10 \Rightarrow  3 x'  \bmod{32} = 5 \Rightarrow x'=23 %\text{~and}\\  --- made overfull box
;\\
3 x' + 5 \bmod{32} = 11 \Rightarrow 3 x'  \bmod{32} = 6 \Rightarrow x'=2.%\hspace*{.55em}
\end{array}$$
It remains to solve for $x$ in  $x'=x\bmod 32$ with the constraint that $x$ is an integer in $[0,64)$.
When $x'=2$, we have that $x\in \{2, 34\}$.
When $x'=23$, we have that $x\in \{23,55\}$.
Hence, the solutions are 2, 23, 34 and 55. 
\end{example}

Using Proposition~\ref{prop:almostinvertible}, we can show that fast variations of Multilinear are strongly universal  even though we use regular integer arithmetic,  not finite fields.

\begin{theorem}
\label{thm:multilinearisstrong}
Given integers $K,L$ satisfying $K\geq L-1\geq 0$, consider the families of $(K-L+1)$-bit hash functions
\begin{itemize}
\item\textsc{Multilinear}:
\begin{eqnarray*}h(s)= \Big (\Big (m_1+ %\textstyle
\sum_{i=1}^{n} m_{i+1} s_i \Big )\mod{2^K}\Big ) \div 2^{L-1}\end{eqnarray*}
\item\textsc{Multilinear-HM}: 
\begin{eqnarray*}h(s)= \Big (\Big (m_1+%\textstyle
\sum_{i=1}^{ n/2} (m_{2i}+ s_{2i-1})(m_{2i+1}+ s_{2i})\Big )\\  \mod{2^K}\Big ) \div 2^{L-1}\end{eqnarray*}
which assumes that $n$ is even. 
\end{itemize}
Here the $m_i$'s are random  integers in $[0,2^K)$ and  the string characters $s_i$ are integers in $[0,2^L)$. 
These two families are strongly universal over fixed-length strings, or over variable-length strings that do not end with the zero character. We can apply the second family to strings of odd length  by appending an extra zero element so that all strings have an even length. 
\end{theorem}
\begin{proof}
We begin with the first family (\textsc{Multilinear}).
Given any two distinct strings $s$ and $s'$, consider
the equations $h(s)=y$ and $h(s')=y'$ for any two hash values $y$ and $y'$.
Without loss of generality, we can assume that the strings have the same length. If not,
we can pad the shortest string with zeros without changing its hash value.
We need to show that $P(h(s)=y \land h(s')=y')= 2^{2(L-K-1)}$.
Because
the two strings are distinct, we can find $j$ such that
$s_j \neq s'_j$. Without loss of generality, assume that
$s'_j -s_j \in [0,2^L)$.

We want to solve the equations 
\begin{eqnarray}
&\left (\left (m_1+\textstyle\sum_{i=1}^{n} m_{i+1} s_i \right ) \mod {2^K}\right )  \div 2^{L-1} =  y,\label{eqn:xx}\\
&\left (\left (m_1+\textstyle\sum_{i=1}^{n} m_{i+1} s'_i \right )\mod {2^K}\right ) \div 2^{L-1}   =  y'\label{eqn:xxx}
\end{eqnarray}
for integers $m_1, m_2, \ldots$ in $[0,2^K)$.

Consider the following equation  
\begin{eqnarray*}
\left (m_1+\textstyle\sum_{i=1}^{n} m_{i+1} s_i \right ) \mod {2^K}   & = & z.
\end{eqnarray*}
There is a bijection
between $m_1$ and $z\in [0,2^K)$. That is, for every value of $m_1$,
there is a unique $z$, and vice versa. Specifically, we have
\begin{eqnarray*}
m_1    & = & z - \sum_{i=1}^{n} m_{i+1} s_i  \mod {2^K}.
\end{eqnarray*}
If we choose $z$ such that $z\div  2^{L-1}= y$, we
effectively solve Equation~\ref{eqn:xx}. 
By substitution in Equation~\ref{eqn:xxx}, we have
\begin{eqnarray*}\begin{split}
&\Big ( m_{j+1} (s'_j -s_j) +z 
 + \textstyle\sum_{i\neq j,i=1}^{n} m_{i+1} (s'_i-s_i)
 \\
&\quad \mod {2^K} \Big )  \div 2^{L-1} 
  =  y'.\end{split}
\end{eqnarray*}
This equation is independent of $m_1$.
 By Proposition~\ref{prop:almostinvertible}, there are
exactly $2^{L-1}$~solutions $m_{j+1}$ to this last equation.
(Indeed, in the statement of Proposition~\ref{prop:almostinvertible}, substitute 
$m_{j+1}$ for $x$, $s'_j -s_j$ for $a$, $z+\sum_{i\neq j, i=1}^{n} m_{i+1} (s'_i-s_i)  \mod {2^K}$ for $c$ and $y'$ for $b$.) 

Meanwhile, there are $2^{L-1}$~possible values  $z$ such that $z \div 2^{L-1} = y$.
Because there is a bijection between $m_1$ and $z$, there are also $2^{L-1}$~possible values for $m_1$.

So, focusing only on $m_1$ and $m_{j+1}$, there are 
$2^{L-1} \times 2^{L-1}$ values satisfying  $h(s)=y$ and $h(s')=y'$. Yet there
are $2^K \times 2^K$~possible pairs $m_1, m_{j+1}$. Thus
the probability that $h(s)=y$ and 
also that %intentional verbosity
$h(s')=y'$ 
 is
$\frac{2^{L-1} \times 2^{L-1}}{2^K \times 2^K}=2^{2(L-K-1)}$, which completes the proof for the first family.

The proof that the second family (\textsc{Multilinear-HM}) is strongly universal is similar. 
As before, set  $z$   in $[0,2^K)$ such that   $z \div 2^{L-1} = y$.
 Solve for $m_1$ from the first equation:
\begin{eqnarray*}
m_1=\Big (z-%\textstyle
\sum_{i=1}^{n/2} (m_{2i}+ s_{2i-1})(m_{2i+1}+ s_{2i})\Big ) \bmod {2^{K}}.
\end{eqnarray*}
Then by substitution, we get 
\begin{eqnarray*}\begin{split}
& \Big(\Big (\textstyle\sum_{i=1}^{n/2} (m_{2i}+ s'_{2i-1})(m_{2i+1}+ s'_{2i})  - \\ 
&\quad (m_{2i}+ s_{2i-1})(m_{2i+1}+ s_{2i}) \\
&\quad \quad +  z \Big ) \mod 2^K  \Big)\div 2^{L-1} =y'.
\end{split}\end{eqnarray*}
%fiddling with wording so that equations that are almost one line wide
%can fit on a line.  Hopefully it does not feel too verbose now.
We can rewrite  this last equation---%, as either     
if $j$ is even, as
%as either 
$((m_{j} (s'_j-s_j) + \rho+z\mod 2^K )\div 2^{L-1} =y'$%if $j$ is even or as
; if $j$ is odd, as 
$((m_{j+2} (s'_j-s_j) + \rho+z\mod 2^K )\div 2^{L-1} =y'$%
%if $j$ is odd
, where $\rho$
is independent of either $m_j$ (when $j$ is even) or $m_{j+2}$ (when $j$ is odd). As before, by Proposition~\ref{prop:almostinvertible}, there are
exactly $2^{L-1}$~solutions for $m_{j}$ ($j$ even) or $m_{j+2}$ ($j$ odd)
if $z$ is fixed. As before, there are $2^{L-1}$~distinct possible values for $z$, and $2^{L-1}$~distinct corresponding values for $m_1$. Hence,
the pair $m_1, m_j$ can take $2^{L-1} \times 2^{L-1}$~distinct values out
of $2^K \times 2^K$ values, which completes the proof.
\end{proof}

To apply  Theorem~\ref{thm:multilinearisstrong} to variable-length strings, we can append the character value one to all strings so that they never end with the character value zero, as in \S~\ref{sec:multilinear}.  

To extend \textsc{Multilinear-HM} to variable-length strings, we can hash variable-length strings as if they were fixed-length strings ($h(s)$). We then also hash the length of the string ($|s|$) using another pairwise independent hash family ($g(|s|)$. Finally, we combine the two hash values using the exclusive~or ($h(x) \oplus g(|s|)$): the result is pairwise independent over variable-length strings. There are more efficient alternatives as well. For example, if the length of the string fits in a word ($|s| < 2^L$), we can prepend to all strings a word containing their length before computing the hash value.

Theorem~\ref{thm:multilinearisstrong} is both more
general (because it includes strings) and more specific (because
the cardinality of the set of hash values is a power of two) than
a similar result by 
Dietzfelbinger~\cite[Theorem~4]{dietzfelbinger1996universal}. However,
we believe our proof is more straightforward: we mostly use
elementary mathematics. 

While Dietzfelbinger did not consider the multilinear case,
others
proposed variations suited to string hashing.
P\v{a}tra\c{s}cu and Thorup~\cite{patrascu2010power} state without proof that  \textsc{Multilinear-HM}
over strings of length two is strongly universal for $K=64,L=32$. 
They extend this approach to strings, taking characters
two by two: \begin{eqnarray*}\begin{split}
h(s)&= \Big (\Big (\bigoplus_{i=1}^{n/2} (m_{3i-2}+ s_{2i-1})(m_{3i-1}+ s_{2i}) + m_{3i }\Big )\\
& \quad \mod{2^K}\Big ) \div 2^{L}\end{split}\end{eqnarray*}
where $\bigoplus$ is the bitwise exclusive-or operation 
and $n$ is even. 
Unfortunately, their approach
uses  more operations  and requires 50\% more random numbers
than  \textsc{Multilinear-HM}. 
 They also refer
to an earlier reference \citep{1496842} where a similar scheme was erroneously
described
as universal, and presented as folklore: 
\begin{eqnarray*}\begin{split}
h(s)&= \Big (\Big(\bigoplus_{i=1}^{n/2} (m_{2i+1}+ s_{2i+1})(m_{2i+2}+ s_{2i+2})\Big)\\&\quad\mod{2^K}\Big ) \div 2^{L}\end{split}\end{eqnarray*}
where $n$ is even.
To falsify the universality of this last family, we can verify numerically that the strings $0,0$ and
$2,6$ collide with probability $\frac{576}{4096}>\frac{1}{2^3}$,
for $K=6,L=3$. %inverted sentence to avoid splitting an expression 
In any case,
we see no benefit to this last approach for long strings because  
\textsc{Multilinear-HM}
is likely just as fast, and it is strongly universal. 

\subsection{Implementing \textsc{Multilinear}}
\label{sec:implementing-multilinear}

If 32-bit  values are required,
we can generate a large buffer of 64-bit 
unsigned 
 random integers $m_i$. 
The computation of either 
\begin{eqnarray*}
h(s)= \Big (m_1+\sum_{i=1}^{n} m_{i+1} s_i \mod{2^{64}}\Big ) \div 2^{32}
\end{eqnarray*}
or
\begin{eqnarray*}\begin{split}
h(s)= \Big (&m_1+\sum_{i=1}^{n/2} (m_{2i}+ s_{2i-1})(m_{2i+1}+ s_{2i}) \\& \quad\mod{2^{64}} \Big ) \div 2^{32}\end{split}
\end{eqnarray*}
is then a simple matter using unsigned integer arithmetic common to most modern processors. 
The division by $2^{32}$ can be implemented efficiently by a right shift (\verb+>>+32).

One might object that according to Theorem~\ref{thm:multilinearisstrong}, 
63-bit random numbers are sufficient if we wish to hash 32-bit characters to a 32-bit hash
 value. The division by $2^{32}$ should then be replaced by a division by $2^{31}$. However,
  we feel that such an optimization is unlikely to either save memory or improve speed. 

\textsc{Multilinear} is essentially an inner product
and thus
can benefit from multiply-accumulate CPU instructions:
by processing the multiplication and the subsequent addition
as one 
machine operation, 
the processor may be able
to do the computation faster than if the computations were
done separately. 
Several processors have such integer multiply-accumulate instructions
 (ARM, MIPS, Nvidia and PowerPC).
Comparatively, we do not know of any 
multiply-xor-accumulate instruction in popular processors.

Unfortunately, some languages---such as Java---fail to support unsigned integers.
With a two's complement 
representation, the de facto standard in modern processors, 
additions and multiplications give identical results, up to overflow
flags, as long as no promotion is involved: e.g., multiplying
32-bit integers using 32-bit arithmetic, or 64-bit integers
using 64-bit arithmetic. 
However,
we must still be careful: promotions
and divisions differ when we use signed integers:
\begin{itemize}
\item If we store string characters using 32-bit integers
(\texttt{int}) and random values as 64-bit integers (\texttt{long}),
Java 
will sign-extend 
the 32-bit integer to a 64-bit integer when computing  
 $\texttt{m}_{i+1} * \texttt{s}_i$,
giving an unintended result for negative string characters.
Use $ \texttt{m}_{i+1} * (\texttt{s}_i \& \texttt{0xFFFFFFFFl})$ instead.

\item Unsigned and signed divisions differ.
Correspondingly, 
for
the division by $2^{32}$---to retrieve the 32~most significant bits---%
the unsigned right-shift operator (\verb+>>>+) must be used in Java, and
not the regular right shift (\verb+>>+).
\end{itemize}

Because we assume that the number of bits is a constant, 
the computational complexity of \textsc{Multilinear}  is linear ($O(n)$). 
\textsc{Multilinear} uses $n$~multiplications, $n$~additions,
and one shift,
whereas \textsc{Multilinear-HM} uses $n/2$~multiplications, 
$3n/2$~additions, and one shift.
In both cases we use $2n+1$~operations, although there may be benefits to having fewer multiplications. 
(Admittedly, Single Instruction, Multiple Data (SIMD) 
processors can do several instructions at once, % which makes 
making
an analysis
based solely on the number of operations misleading.)

Consider that we need at least $\approx 32 (n+1)$~random bits for strongly
universal 32-bit hashing of $32n$~bits~\citep{188765} according to Stinson's bound.
That is, we must aggregate $\approx 64 n + 32$~bits into a  32-bit hash value.
Assume that we only allow unary and binary operations. 
A 32-bit binary operation maps 64~bits to 32~bits, a reduction
of 32~bits. Hence, we require at least $2n$~32-bit operations for strongly universal hashing.
Alternatively, we require at least $n$~64-bit operations.
Hence, for $n$~large, \textsc{Multilinear} and \textsc{Multilinear-HM} use
at most twice the minimal number of operations.

\subsection{Word size optimization}
\label{sec:word-size}

The number of required bits is application dependent: for a hash table, one may be able to bound
the maximum table size.  In several languages such as Java, 32-bit hash values
are expected.
Meanwhile the key parameters of our hash functions \textsc{Multilinear} and \textsc{Multilinear-HM} are $L$ 
(the size of characters)
and $K$ (the size of the operations), and these two hash functions deliver $K-L+1$~usable bits.

However, both $K$ and $L$ can be adjusted given a desired number of usable random bits.
Indeed, a length $n$ string of $L$-bit characters can be reinterpreted as a 
length $n \lceil L/L' \rceil $
string
of $L'$-bit characters, for any non-zero $L'$. Thus, 
we can either grow $L$ and $K$ or reduce $L$ and $K$, for the same number of usable bits.

To reduce the need for random bits, 
we should
use large values of $K$.
 Consider a long input string that we can represent as a string of
32-bit or 96-bit characters.
 Assume we want 32-bit hash values. Assume also that our random data only comes
 in strings of  64-bit or 128-bit characters.  
 If we process the string as  a 32-bit string, we require 64~random bits per character.
 The ratio of random strings to hashed strings is two. 
 If we process the string as a 96-bit string, we require 128~random bits per character
 and the ratio of random strings to hashed strings is $128/96=4/3\approx 1.33$.
 What if we could represent the string using 224-bit characters and have random
 bits 
packaged into
characters of 256~bits? We would then have a ratio of $8/7\approx 1.14$.
 
We can formalize this result.
Suppose we require $z$~pairwise independent bits and that we have $M$~input bits. 
Stinson~\cite{188765} showed that this requires at least $1+2^M(2^z-1)$~hash functions. %be verbose to avoid splitting exressions     or, equivalently, 
Equivalently, this requires
$\log(1+2^M(2^z-1))$~random bits. Thus, given any hashing family,
the ratio of its required number of random bits to the Stinson limit (henceforth Stinson
 ratio) must be greater or equal to one. 
The $M$~input 
bits can be represented as an $L$-bit $n$-character string for $M=nL$. Under
\textsc{Multilinear} (and \textsc{Multilinear-HM}), we must have $z=K-L+1$. Thus
we use $K(n+1)=(z+L-1)(\lceil M /L \rceil +1)$~random bits. We have that 
  $(z+L-1)(\lceil M /L \rceil +1) \leq (z+L-1)( M /L +2)$ which is minimized when 
\begin{equation} L=\sqrt{(z-1) \frac{M}{2}}.
\label{eqn:L-minimizing-randbits}
\end{equation}
  Rounding $L=\sqrt{(z-1) M/2}$ up and substituting it back into $(z+L-1)(\lceil M /L \rceil +1)$,
  we get an upper bound on the number of random bits required by  \textsc{Multilinear}. 
This bound
is nearly optimal when $\lceil M /L \rceil \approx M/L$, that is, when $M$ is large. 
  Unfortunately, this estimate fails to consider 
that word sizes are usually prescribed.
For example, we could be required to choose
  $K\in \{8,16,32,64\}$. That is, we have to choose $L\in\{ 9-z,17-z,33-z,65-z \}$.
  Fig.~\ref{fig:stinsonlimit} shows the corresponding Stinson  ratios.
   When there
  are many input bits ($M\gg 1$), the  ratio of  \textsc{Multilinear} converges
  to one. That is, as long as we can decompose input data into strings of large
  characters (having %$approx
approximately %avoid splitting an expression
$\sqrt{(z-1) M/2}$~bits), \textsc{Multilinear} requires
  almost a minimal number of bits. This may translate into an optimal memory usage.
  (The result also holds for \textsc{Multilinear-HM} except that it is slightly less
  efficient for strings having an odd number of characters.)
 If we restrict the word sizes  to common machine word sizes ($K\in \{8,16,32,64\}$), the 
 ratio is $\approx 2$ for large input strings. 
  We also consider the case where we could use 128-bit words (with $K\in \{8,16,32,64, 128\}$). It
  improves
 the ratio noticeably ($\approx 1.33$), as expected. 
 
\begin{figure}
\centering
\includegraphics[width=0.95\columnwidth]{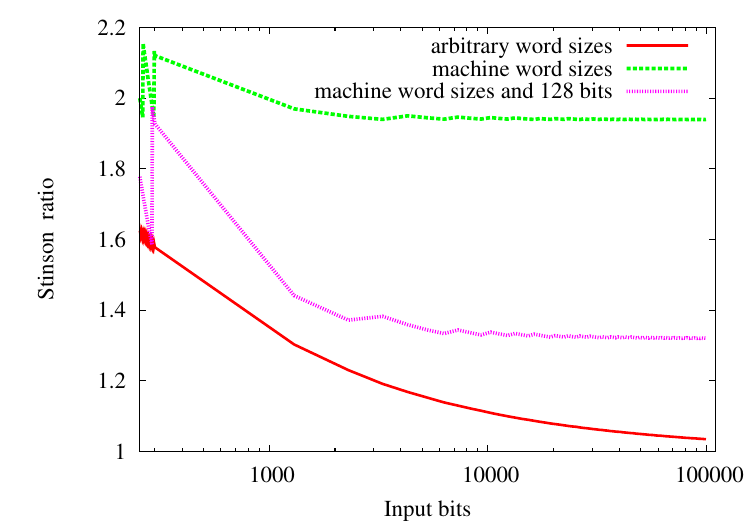}
\caption{\label{fig:stinsonlimit} For large inputs, \textsc{Multilinear} 
requires an almost optimal number of random bits when arbitrary word sizes ($K$) are allowed.
It has lower efficiency when the word size is constrained. 
The plot was generated for 32-bit hash values ($z=32$).
}
\end{figure}
 
We can also choose the word size ($K$) to optimize speed. On a 64-bit processor, setting $K=64$ would make sense.  
We can compare this default with two  
alternatives: 
\begin{enumerate}
 \item We can try to support much larger words using
fast multiplication algorithms such as Karatsuba's.
We could merely try to minimize the number of random bits. However, this ignores the growing
computation cost of multiplications over many bits,
 e.g., 
 Karatsuba algorithm is in % "Karatsuba algorithm" is much more common on Google Scholar than Karatsuba's algorithm
$\Omega(K^{1.58})$.   %was n^1.58 but we have used K before and after this...
 For simplicity,
 suppose that the cost 
 of $K$-bit multiplication 
 costs $K^a$ time for $a>1$.
To hash $M$~bits,
 we require $\lceil M/L \rceil $~multiplications with \textsc{Multilinear}.
 When we have long strings (i.e., $M\gg L$), we  %added i.e. to be verbose
 can simplify $\lceil M/L \rceil \approx M/L$.
 If we desire $z$-bit hash values, 
 then we need to use multiplication on $K=z+L-1$~bits.
 Thus, the processing cost can be (roughly) approximated as
 $\frac{M (z+L-1)^a}{L}$. 
 Starting from $L=1$, this function
 initially decreases to a minimum at 
\begin{equation} L=\frac{z-1}{a-1}
\label{eqn:L-minimizing-cost}
\end{equation} 
  before increasing again as $L^{a-1}$.
 (When $a=1.5$ and $z=32$, we have $\frac{z-1}{a-1}=62$.)
 See Fig.~\ref{fig:multiplicationcost}.
Hence, while we can  minimize
 the total number of random bits by using many bits per character ($L$ large),
 we may want to keep $L$ relatively small to take into account the superlinear
 cost of multiplications.

 \item We can support  128-bit words on a
 64-bit processor, with  some overhead. 
(Recent GNU GCC~compilers have the \texttt{\_\_uint128} type, as a C-language extension.)  A single 128-bit multiplication may require up
 to three 64-bit multiplications. However, it processes more data: with $z=32$ hashed bits, 
 each 128-bit multiplication hashes 97~input bits. 
 Comparatively, setting
 $K=64$, we require a single 64-bit multiplication, but we process only 
 32~bits of data. 
(Formally, we could process 33~bits of data, but
 for convenient implementation, we process data in powers of two.)
 Hence, it is unclear which approach is faster: three 64-bit multiplications and 128~bits of
 random data to process 97~input bits, or a single 64-bit multiplication and 64 bits~of random data,
 to process  33~input bits. However, the 128-bit approach will use 33\% fewer random bits.
 Going to 256-bit word sizes would only reduce the number of random bits by 14\%: using
 larger and larger words leads to diminishing returns.
 \end{enumerate}

We assess these two alternatives experimentally in \S~\ref{sec:experimentswithgmp}.

\begin{figure}
\centering
\includegraphics[width=0.95\columnwidth]{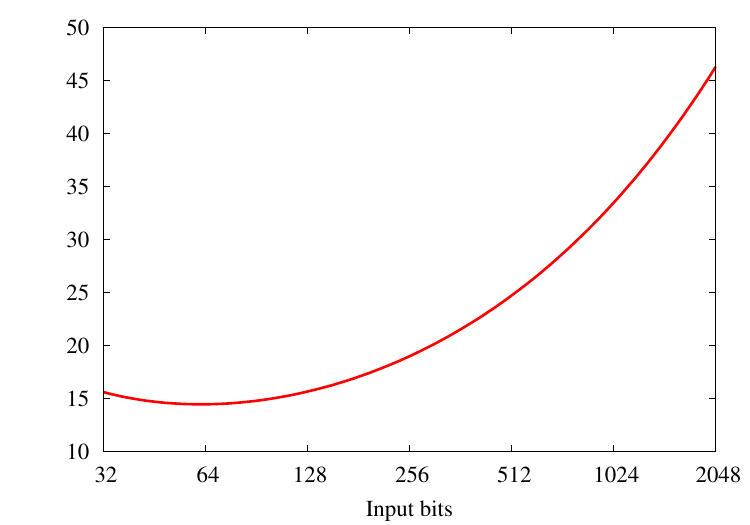}
\caption{\label{fig:multiplicationcost} Modeled computational
cost per bit as a function of the number of bits per character ($\frac{(z+L-1)^a}{L}$)
for 32-bit hashing values ($z=32$) and $a=1.5$. 
}
\end{figure}

\section{Fast Multilinear with carry-less multiplications}
\label{sec:carryless}

To help support fast operations over binary finite fields ($GF(2^L)$), 
AMD and Intel introduced the Carry-less Multiplication (CLMUL) instruction set~\citep{Gueron2010549}. 
%Given   intentional verbosity
If we are given
the binary representations of two numbers, $a=\sum_{i=1}^L a_i 2^{i-1}$ and $b=\sum_{i=1}^L b_i 2^{i-1}$, the carry-less multiplication is given
by $c=\sum_{i=1}^{2L-1} c_i 2^{i-1},$ where $c_i = \bigoplus_{j=1+i}^{2L-1} a_j b_{j-i}$.
Henceforth, we write $a \star b = c$.
If we represent the two $L$-bit integers $a$ and $b$ as polynomials in
GF(2)[$x$], then the carry-less multiplication is equivalent to the usual
polynomial multiplication:
\begin{eqnarray*}
\left (\sum_{i=1}^L a_i x^{i-1}\right )\left (\sum_{i=1}^L b_i x^{i-1}\right ) = \sum_{i=1}^{2L-1} c_i x^{i-1}.
\end{eqnarray*}

With a fast carry-less computation, we can compute Multilinear efficiently. Given any irreducible polynomial $p(x)$ of degree $L$, the field GF(2)[$x$]/$p(x)$ is isomorphic to GF$(2^L)$.
Hence, we want to compute $
h(s)=m_1+\sum_{i=1}^n m_{i+1} s_i$ over GF(2)[$x$]/$p(x)$. 
%(Similarly, we can use 
%Equation~\ref{mhm-gf} to reduce the number of multiplications by half: we discuss this possibility at the end of this section.)
Computing all multiplications over GF(2)[$x$]/$p(x)$ would still be expensive given fast carry-less multiplication. Instead,  we first compute $m_1+\sum_{i=1}^n m_{i+1} s_i$ over GF(2)[$x$] and then return the remainder of the 
division of the
final result by $p(x)$. 
Indeed, think of the values $m_1, m_2, \ldots$ and $s_1, s_2, \ldots$ as polynomials of
degree at most $L$ in GF(2)[$x$]. Each of the $n$~multiplications in GF(2)[$x$] is equivalent to a carry-less multiplication over $L$-bit integers 
which results in a $2L-1$-bit value. Similarly, each of the $n$~additions in GF(2)[$x$] is an exclusive-or operation.
That is, we want to compute the $2L-1$-bit integer  \begin{eqnarray}\bar h(s)=m_1 \oplus\left ( \bigoplus_{i=1}^n  m_{i+1}\star s_i \right ).\label{eqn:gfmultilinear}\end{eqnarray}
Finally, considering $\bar  h(s)$ as an element of GF(2)[$x$], noted $q(x)$, we must compute $q(x)/p(x)$. The remainder (as an $L$-bit integer) is the final hash value $h(s)$.

If done naively, computing the remainder of the division by an irreducible polynomial
may remain relatively expensive, especially for short strings since they require few multiplications. 
A common technique to quickly compute the remainder is the Barrett reduction algorithm~\cite{Barrett:1987:IRS:36664.36688}. The adaptation of this reduction to 
GF(2)[$x$] is especially convenient~\cite{springerlink:10.1007/978-3-540-69499-1_7}
when we choose the irreducible polynomial $p(x)$ such that
$\textrm{degree}(p(x)-x^L)\leq L/2$, that is, when we can
write it as $p(x)= \sum_{i=0}^{\lfloor L/2 \rfloor} a_i x_i + x^L$. 
(There are such irreducible polynomials for $L\in\{1,2,\ldots,400\}$~\cite{irredpolyonline}
and we conjecture that such a polynomial can be found for any $L$~\citep{springerlink:10.1007/978-3-540-24633-6_16}.)
In this case, the remainder of $q(x)/p(x)$ is given by 
\begin{eqnarray*}
((((q \div 2^L) \star p ) \div 2^L ) \star p) \oplus q ) \mod 2^L
\end{eqnarray*}
 where $q$ and $p$ are the $2L-1$-bit and $L+1$-bit integers representing $q(x)$ and $p(x)$.
 (See~\ref{appendix:implementationscarryless} for implementation details.)  We expect the two  carry-less multiplications  to account for most of the running time of the reduction.
 Yet we expect that even a fast implementation
of the Barrett reduction is much slower than merely
selecting the left-most $L$~bits as in \textsc{Multilinear}.

Unfortunately, in its current  Intel implementation, carry-less multiplications
have significantly reduced throughput compared to 
regular integer multiplications. Indeed,  with pipelining, it is
possible to complete one regular multiplication per cycle, but
only one carry-less multiplication every 8~cycles~\citep{inteloptimization}.
However, using a result from \S~\ref{sec:multilinear}, we can reduce the number of multiplications by half if we compute 
 \begin{eqnarray*}\bar h(s)=m_1 \oplus \left ( \bigoplus_{i=1}^{n/2}  (m_{2i} + s_{2i-1}) \star (m_{2i+1} + s_{2i})\right )\end{eqnarray*}
instead. 
(Henceforth, we refer to this last variation as \textsc{GF Multilinear-HM},
whereas
we refer to  the version based on Equation~\ref{eqn:gfmultilinear} as \textsc{GF Multilinear}.)

However, irrespective of its speed, the carry-less approach might still be preferable to the schemes described in \S~\ref{sec:technical} (e.g., \textsc{Multilinear}) because fewer random bits are required. Indeed,  to generate $L$-bit hash values from $n$-character strings, the carry-less scheme uses $(n+1)L$~random bits, whereas \textsc{Multilinear} requires $2L + n (2L-1)$~random bits. 

\section{Experiments}
\label{sec:experiments}

Our experiments show the following results:
\begin{itemize} 

\item It is best to implement \textsc{Multilinear}  with loop unrolling.
With this optimization,  \textsc{Multilinear} is just
as fast (on  Intel processors) as  \textsc{Multilinear-HM}, even though it has twice the number
of multiplications. 
 In general, processor microarchitectural
differences are important in determining which
method is fastest.
(\S~\ref{sec:mainexperiments})

\item In the absence of processor support for carry-less multiplication (see \S~\ref{sec:carryless}), hashing using Multilinear
 over binary finite fields is an order of magnitude slower than \textsc{Multilinear}
even when using a highly optimized library.  
(\S~\ref{sec:experimentsonfinitefields})

\item  Even with hardware support for carry-less multiplication, hashing using Multilinear
 over binary finite fields remains several times slower than \textsc{Multilinear}.  
(\S~\ref{sec:carryless-experiments})

\item Given a 64-bit processor, it is noticeably faster to use a word size of 64~bits even though a larger
word size (128~bits) uses fewer random bits (33\% less). 
Use of  multiprecision arithmetic libraries can further reduce the overhead from accessing random bits,
but they also fail to be  competitive with respect to speed, though they can 
halve 
the number of required random bits. 
 (\S~\ref{sec:experimentswithgmp})

\item 
\textsc{Multilinear} is generally faster than popular string-hashing algorithms.
(\S~\ref{sec:revisit})

\end{itemize}

\subsection{Experimental setup}

We  evaluated the  hashing functions on the platforms shown in Table~\ref{tbl:platforms}.
Our software is freely available online~\cite{oursoftware}. 
For Intel and AMD processors, we used the processor's time stamp counter (\texttt{rdtsc} instruction~\cite{intelbenchmark})
to estimate the number of cycles required to hash each
 byte. Unfortunately, the ARM instruction set
does not provide access to such a counter. Hence, 
for ARM processors (Apple~A4 and Nvidia Tegra),
we estimated the number of cycles required by dividing the wall-clock time
by the documented processor clock rate (1\,GHz).

\begin{table*}\caption{\label{tbl:platforms}Platforms used.}\centering
\begin{minipage}{1.6\columnwidth} %needed for footnotes
\centering
\singlespacing
%\small
\begin{tabular}{llll}\hline
Processor & Bits & GCC version & Flags, besides -O3 -funroll-loops\\ \hline
\multicolumn{4}{l}{64-bit processors} \\
\hline
Intel Core 2 Duo      & 64 &GNU GCC 4.6.2 
                                               & -march=core2  -mno-sse2
\\
Intel Xeon X5260      & 64 & GNU GCC 4.1.2     & -march=nocona\\
Intel Core i7-860 & 64 & GNU GCC 4.6.2     & -march=corei7 -mno-sse2 \\
Intel Core i7-2600 & 64 & GNU GCC 4.6.3     & -march=corei7-avx -ftree-no-vectorize\\% -march=corei7 -mno-sse2
Intel Core i7-2677M
                   & 64 & GNU GCC 4.6.2     & -march=corei7 -mno-sse2 \\
AMD Sempron 3500+     & 64 & GNU GCC 4.4.3     & -march=k8 -mno-sse2\\   
AMD  V120 
              & 64 & GNU GCC 4.4.3     & -march=amdfam10 -mno-sse2\\
AMD FX8150 & 64 &  GNU GCC 4.6.3     & -march=bdver1 -ftree-no-vectorize\\
\hline
\multicolumn{4}{l}{32-bit processors} \\
\hline
Intel Atom N270       & 32 & GNU GCC 4.5.2     & -march=atom\\ 
Apple A4              & 32 & GNU GCC 4.2.1     & -march=armv7  \\
Nvidia Tegra~2        & 32 & GNU GCC 4.4.3\footnote{From the Android NDK, configured for the android-9 platform, and used on a Motorola XOOM.}     & \\
VIA Nehemiah          & 32 & GNU GCC 3.3.4     & -march=i686 \\ 
\hline
\end{tabular}
\end{minipage}
\end{table*}

For the 64-bit machines, 64-bit executables were produced and all
operations were executed using 64-bit
arithmetic except where noted. 
All timings were repeated three times.  
For the 32-bit processors, 
 32-bit operations were used to process 16-bit
strings.  Therefore, results between 32- and 64-bit processors are not
directly comparable.
Good optimization flags were found by a trial-and-error process.  We note that using profile-guided optimizations
did not improve this code any more than simply enabling loop unrolling (\mbox{\texttt{-funroll-loops}}).  With (only) versions 4.4 and higher
of GCC, it was sometimes important for speed  to forbid use of SSE2 instructions when compiling \textsc{Multilinear} and \textsc{Multilinear-HM} (hence the \texttt{-mno-sse2} flags in Table~\ref{tbl:platforms}). Moreover, we determined that  versions 4.6 and 4.7 of GCC gave incorrect compiled code when vectorizing \textsc{Multilinear} and related functions: as an alternative to the \texttt{-mno-sse2} flag, we found that the \texttt{-fno-tree-vectorize} flag was sufficient to ensure %insure
correct results.   
% owen: while most style geeks seem to be okay using 'insure' as a synonym
% for 'assure', some say we'd want 'ensure'.  Since we ensure elsewhere, I
% changed this...

We found that the speed 
is insensitive to the content of the string: in our tests we hashed randomly generated strings.
We reuse the same string for all tests. Unless otherwise
specified, we hash randomly generated
32-bit strings of 1024~characters. 

In addition to \textsc{Multilinear} and
\textsc{Multilinear-HM}
we also implemented \textsc{Multilinear}~(2-by-2) which is
merely \textsc{Multilinear} with 2-by-2 loop unrolling.
(See~\ref{appendix:implementations}
for representative C implementations.)

Our timings should 
represent the best possible performance: the performance
of a function may degrade~\citep{Greenangalois} when it is included in an application because
of bandwidth and caching.

\subsection{Reducing the multiplications or unrolling may fail to improve the speed}
\label{sec:mainexperiments}
We ran our experiments over both the 32-bit and 64-bit processors.
For the 32-bit processors, we generated both 16-bit and 32-bit hash values.
Our experimental results are given in Table~\ref{table:results}.

\begin{table*}
\centering 
\begin{minipage}{1.6\columnwidth}
\centering
\singlespacing \caption{\label{table:results}Estimated CPU cycles per byte for fast Multilinear hashing
}
\begin{tabular}{c|ccc}\hline
                   & \textsc{Multilinear} &   2-by-2 & \textsc{Multilinear-HM}\\ \hline
\multicolumn{4}{l}{64-bit processors and 32-bit hash values and characters} \\
\hline
Intel Core 2 Duo
                  &           0.54        &    \textbf{0.52}   &        \textbf{0.52} \\

Intel Xeon X5260   &      \textbf{0.50}             &    \textbf{0.50}   &       \textbf{0.50} \\
Intel Core i7-860    &      \textbf{0.42}            &    \textbf{0.42}   &        \textbf{0.42} \\
Intel Core i7-2600    &      0.35            &    \textbf{0.27}   &        0.28 \\ 
Intel Core i7-2677M   &      0.25            &  \textbf{0.20} & \textbf{0.20} \\
AMD Sempron 3500+
                   &      0.63             &    0.60  &        \textbf{0.40} \\
AMD V120
                   &      0.63             &    0.63  &        \textbf{0.40}\\
AMD FX8150        &    0.88     
  &    1.00      &  \textbf{0.51}      \\   
\hline
\multicolumn{4}{l}{64-bit arithmetic and 32-bit hash values and characters on 32-bit processors} \\
\hline
Intel Atom N270   &                4.2  &      4.2                & \textbf{3.6}    \\
Apple A4           &      3.0                &  \textbf{2.7}        & 3.3           \\ 
Nvidia Tegra~2   &         3.3                 &  \textbf{3.0}       &  4.9           \\   
VIA Nehemiah     &        12                 &   12               &   \textbf{8.2}\\
\hline
\multicolumn{4}{l}{32-bit processors and 16-bit hash values and characters} \\
\hline
Intel Atom N270     & \textbf{2.1}  & 3.5 & 2.6\\
Apple A4            &      1.9                 &  2.6        & \textbf{1.7}           \\ 
Nvidia Tegra~2   &         \textbf{1.8}                 &  2.2       &  1.9           \\   
VIA Nehemiah      &       5.2          &   5.2  &        \textbf{3.6} \\
\hline
\end{tabular}
\end{minipage}
\end{table*}

 We see that
over 64-bit Intel processors (except for the i7-2600), there is little difference between \textsc{Multilinear}, \textsc{Multilinear}~(2-by-2)
and \textsc{Multilinear-HM}, even though \textsc{Multilinear} and \textsc{Multilinear}~(2-by-2) have twice
the number of multiplications. We believe that Intel processors use aggressive
pipelining techniques well suited to these computations. 
On the AMD processors,
\textsc{Multilinear-HM} is the clear winner, being at least 33\% faster.

For the 32-bit processors, we get vastly different results depending
on whether we generate 16-bit or 32-bit hash values. 
\begin{itemize}
\item As expected, it is roughly
twice as expensive to generate 32-bit hash values than to generate 16-bit values. 
\item For
the VIA processor, \textsc{Multilinear-HM} is 45\% faster than 
\textsc{Multilinear} and \textsc{Multilinear}~(2-by-2). We suspect
that the computational cost is tightly tied to the number of multiplications. 
\item When the 32-bit ARM-based processors generate 32-bit hash values,
\textsc{Multilinear}~(2-by-2) is preferable. 
We are surprised that 
\textsc{Multilinear-HM} is the worse choice. We believe
that this is related to the presence of a multiply-accumulate
instruction in ARM processors.
When  generating 16-bit hash values, \textsc{Multilinear}~(2-by-2) becomes the worse
choice. There is no significant benefit to using \textsc{Multilinear-HM} as opposed
to \textsc{Multilinear}. 
\item The Intel Atom processor benefits from \textsc{Multilinear-HM} when generating 32-bit 
hash value, but  \textsc{Multilinear} is preferable to generate 16-bit hash values.
As with 
the ARM-based processors,  \textsc{Multilinear}~(2-by-2) is a poor choice 
 for generating 16-bit hash values.
\end{itemize}

\subsection{Binary-finite-field 
libraries are not competitive}

\label{sec:experimentsonfinitefields}

We obtained the $\texttt{mp}\mathbb{F}_{b}$ library 
from INRIA\@.
This code is reported~\cite{mpfb} to be generally faster than popular alternatives
such as NTL 
and Zen, and 
our own tests 
found it to be more than twice as fast as Plank's library~\cite{plank2007fast}. 

We computed Multilinear in $GF(2^{32})$, using the version with half the number
of multiplications (see Equation~\ref{mhm-gf}) 
because the library does much more work in multiplication than addition.  Even so,
on a Core~2~Duo, hashing 32-bit strings of 1024~characters was an order of magnitude slower
than \textsc{Multilinear}: 
averaged over a million attempts, the code using  $\texttt{mp}\mathbb{F}_{b}$ required 
an average of 7.69\,$\mu s$ per string, 
compared with 0.78\,$\mu s$ 
for \textsc{Multilinear}. 
While  our implementation of  \textsc{Multilinear} uses twice as many random bits as Multilinear in $GF(2^{32})$, this gain is 
offset 
by the memory usage of
the finite-field library.

\subsection{Hardware-supported carry-less multiplications are not fast enough}
\label{sec:carryless-experiments}

Intel reports a throughput of  one carry-less product every 8~cycles~\citep{inteloptimization} on a processor such as the Intel Core i7-2600. Consider
\textsc{GF Multilinear-HM}: it uses one carry-less multiplication for every
two 32-bit characters. Hence, it requires at least 4~cycles to process each 
character. Therefore, in the best scenario possible,
\textsc{GF Multilinear-HM} will  be almost four   %four  -- why wasn't this more than 4 earlier: ie 8? % Too many numbers to keep straight.
 times slower than
\textsc{Multilinear-HM}  which requires only 
%0.5~cycle per 32-bit character (0.13~cycle per byte).     - old buggy number
1.1~cycles per 32-bit character (0.28~cycle per byte).

To assess the actual performance, we implemented 
both \textsc{GF Multilinear} and \textsc{GF Multilinear-HM} in C (\S~\ref{appendix:implementationscarryless}). 
We also implemented a variation on \textsc{GF Multilinear-HM} (henceforth \textsc{GF Multilinear-HM-Fast}) that loads data in blocks of four 32-bit integers.

\begin{itemize}
\item Of the Intel processors we tested, only the i7-2600 has support for the 
CLMUL instruction set.  
If we use  the flags
\texttt{-O3} \texttt{-funroll-loops} \texttt{-corei7-avx},
we get 9.6~CPU cycles per 32-byte character with \textsc{GF Multilinear}, 5.8~CPU cycles with \textsc{GF Multilinear-HM} and only 4.3~CPU cycles with \textsc{GF Multilinear-HM-Fast}.
That is 2.4~cycles, 1.5~cycles and 1.1~cycles per byte respectively: about 4--9~times slower than \textsc{Multilinear-HM} on the same platform (0.27~cycle per byte).
We might be able to improve our implementation. %: e.g., we expect that much time is spent 
%loading  data into XMM registers. 
However, the throughput of the carry-less multiplication limits the character throughput of \textsc{GF Multilinear} and \textsc{GF Multilinear-HM} to 8 and 4~cycles.
\item One of our AMD processors (AMD FX8150) also supports the CLMUL instruction set. 
 With the flags \texttt{-O3} \texttt{-funroll-loops} \texttt{-march=bdver1}, it fares slightly better than the Intel counterpart: 6.8~CPU cycles per 32-byte character with \textsc{GF Multilinear}, 4.1~CPU cycles with \textsc{GF Multilinear-HM}  and only 3.5~CPU cycles with \textsc{GF Multilinear-HM-Fast}. 
That is 1.7~cycles, 1.0~cycle and 0.9 cycle per byte  respectively. 
 However, the same AMD processor can process each 32-bit character in 2.05~cycles (0.51~cycle per byte) with \textsc{Multilinear}~(2-by-2). Hence,  \textsc{Multilinear} (2-by-2) is 
%3.5~times : Owen: now it is surely 4.2 / 2.05 which is about 2.    
almost 2~times faster than the best carry-less approach (\textsc{GF Multilinear-HM-Fast}).
\end{itemize}
Overall, the hardware-supported carry-less Multilinear schemes are several times slower.
On the bright side, \textsc{GF Multilinear} and \textsc{GF Multilinear-HM} require half the number of random bits.

\begin{figure}
\begin{centering}
\includegraphics[width=.95\columnwidth]{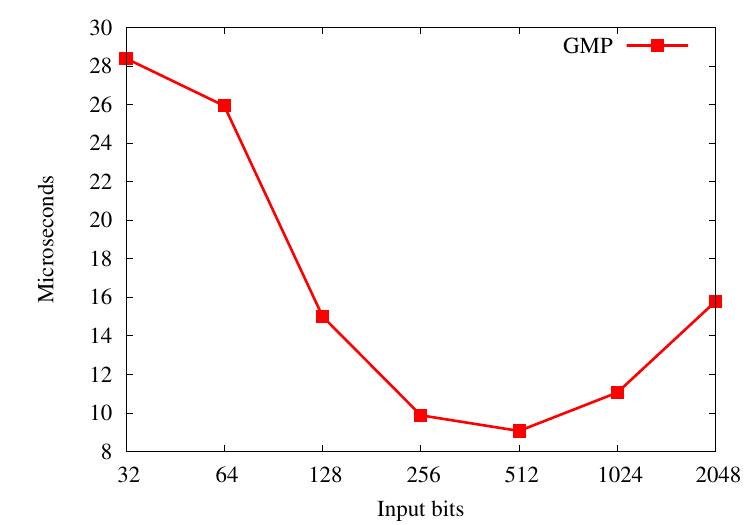}\\
\end{centering}
\caption{\label{fig:stinson-vs-multoverhead} Microseconds to hash 4\,kB using various word sizes and GMP. 
}
\end{figure}

\subsection{The sweet-spot for multiprecision arithmetic is not sweet enough}
\label{sec:experimentswithgmp}
To implement the techniques of \S~\ref{sec:word-size}, we used the GMP
library~\citep{gmpwebsite} version~5.0.2 to
implement \textsc{Multilinear}~(2-by-2).  As usual, we are hashing 4\,kB
of data, though data to be hashed 
are read in large chunks (up to
2048 bits).  The hash output is always 32~bits ($z=32$).  Results show
a benefit as the chunk size $L$ goes from 32 to 512 bits, but thereafter
the situation degrades.  See Fig.~\ref{fig:stinson-vs-multoverhead}.
In the best
case, using 512-bit arithmetic, we require almost 13\,$\mu$s per
string on a Core~2~Duo platform. 
For comparison, we find that the fewest
random bits would be needed when $L=1024$ (\S~\ref{sec:word-size}).
As expected, the running time is minimized for a lower value of
$L$ to account for the superlinear % running -- seems implicit
cost of multiplication.

Unfortunately, we can do 12~times better without the GMP library
(0.78\,$\mu s$ for 64-bit \textsc{Multilinear}), 
so it is not practical to use 512-bit arithmetic---even though it uses fewer random bits (nearly half as many).

%fig 3 was here

\begin{table*}%[tb]  -- style does not seem to accept this
\centering
\begin{minipage}{1.1\columnwidth} %need minipage for footnotes
\singlespacing
\caption{\label{table:results2}A comparison of estimated CPU cycles per byte between fast Multilinear hashing and common hash functions
}
\begin{tabular}{c|ccc}\hline
                   & Rabin-Karp  & SAX   & best Multilinear\\ \hline
\multicolumn{4}{l}{32-bit hash values and characters  on 64-bit processors}\\
\hline
Intel Core 2 Duo
                   &       1.3   & 1.3  &     \textbf{0.52}                    \\
Intel Xeon X5260   &       1.4   & 1.6   &     \textbf{0.50}                  \\
Intel Core i7-860       &       1.4   & 1.6   &     \textbf{0.42}                     \\
Intel Core i7-2600       &       0.89   & 1.1   &     \textbf{0.27}                     \\
Intel Core i7-2677M   &      0.64            &  0.82 & \textbf{0.20} \\
AMD Sempron 3500+  &       1.0   & 1.5   &     \textbf{0.40}                  \\
AMD  V120          &       1.0   & 1.5   &     \textbf{0.40}                     \\
AMD FX8150       &        0.86        &  1.3      & \textbf{0.51} \\
\hline
\multicolumn{4}{l}{32-bit hash values and characters on 32-bit processors} \\
\hline
Intel Atom N270      &   \textbf{1.1}    & 2.0     &       4.2    \\
Apple     A4         &   \textbf{0.88}    & 1.2     &       2.7                      \\ 
Nvidia Tegra~2       &  \textbf{0.85}     & 1.2   &     3.0                        \\ 
VIA Nehemiah           &  \textbf{2.0}   & 3.0     &    8.2                     \\
\hline
\multicolumn{4}{l}{16-bit hash values and characters on 32-bit processors} \\
\hline

Intel Atom N270        &  \textbf{2.1}& 4.1      &    2.2                               \\

Apple     A4           &       \textbf{1.8}   & 2.1      &     \textbf{1.8}                      \\ 
Nvidia Tegra~2         & \textbf{1.6}& 2.4      &     1.7           \\ 
VIA Nehemiah           &       5.0   & 6.6     &     \textbf{3.6}                     \\\hline
\end{tabular}
\end{minipage}
\end{table*}

As a lightweight alternative to a multiprecision library, we
experimented with the \texttt{\_\_uint128} type provided as a GCC extension for
64-bit machines.  We used 128-bit random numbers and processed three
32-bit words with each 128-bit operation.  Since \texttt{\_\_uint128}
multiplications are more expensive than \texttt{\_\_uint128} additions, we
tested the \textsc{Multilinear-HM} scheme.  On the Core~2~Duo machine, the result was
38\% slower than \textsc{Multilinear}~(2-by-2) using 64-bit operations.
This poor results is mitigated by the fact that we use 128~random bits per 96~input bits, versus 64~random bits per 32~input bits (a saving of nearly 33\% for long strings).
Investigation using hardware performance counters showed many
``unaligned loads'' from retrieving 128-bit quantities when we step
through memory with 96-bit steps.  
To reduce this, we tried processing 
only two 32-bit words with each 128-bit operation, since we
retrieved aligned 64-bit quantities.  However, the result was 
61\% slower than \textsc{Multilinear}~(2-by-2) using 64-bit operations.

\subsection{Strongly universal hashing is inexpensive?}
\label{sec:revisit}

In a survey, Thorup~\cite{338597} concluded that strongly universal hash families
are just as efficient, or even more efficient, than popular
hash functions with weaker  theoretical guarantees.
However, he only considered 32-bit integer inputs. We consider strings.

In Table~\ref{table:results2}, we compare the fastest Multilinear (\textsc{Multilinear-HM})
with two non-universal fast 32-bit string hash functions, Rabin-Karp~\cite{karp1987erp}
and SAX~\cite{ramakrishna1997performance}.
(They are similar to hash functions found in 
programming languages such as Java or Perl.)
 Even though these functions were designed
for speed and lack strong theoretical guarantees, they are far slower than \textsc{Multilinear} on desktop processors (AMD and Intel).
Only for ARM processors (Apple     A4  and Nvidia Tegra~2) with 32-bit hash values are they much faster. We suspect that this good result on ARM processors is due 
to the  multiply-accumulate instruction. Clearly, such a 
multiply-accumulate operation greatly benefits simple hashing functions such as Rabin-Karp and SAX. 

%table 3 had been here

Crosby and Wallach~\cite{Crosby:2003:DSV:1251353.1251356} showed that almost universal hashing could be as fast as common deterministic hash functions.
One of their most competitive almost universal
 schemes is
due to Black et al.~\cite{703969}. 
Their  
fast family of hash functions 
is
called NH: 
\begin{eqnarray*}\begin{split}
h(s)=  & \sum_{i=1}^{n/2} \left  (m_{2i-1}+ s_{2i-1} \mod {2^{L/2}} \right )\\ & \quad \times \left (m_{2i}+ s_{2i} \mod {2^{L/2}} \right) \mod{2^L}. \end{split}
\end{eqnarray*}
NH is almost universal over fixed-length strings, or over variable-length strings that do not end with the zero character; we can apply it to strings having odd length by appending a character with value zero. 
It fails to be uniform: the value \begin{eqnarray*}(m_{1}+ s_{1} \mod {2^{L/2}})(m_{2}+ s_{2} \mod {2^{L/2}}) \end{eqnarray*} is zero whenever either $m_{1}+ s_{1} \mod {2^{L/2}}$
is zero %intentionally verbose
 or $m_{2}+ s_{2} \mod {2^{L/2}}$ is zero, which occurs with probability $\frac{2^{L/2+1}-1}{2^{L}}>\frac{1}{2^{L}}$ over all possible values of $m_1,m_2$. 
Moreover, the least significant bits may fail to be almost universal: e.g., for $L=6$, there are 96~pairs of distinct strings 
colliding with probability~1 over the least two significant bits. 
When processing 32-bit characters, it
generates 64-bit hash values with collision
probability of $1/2^{32}$.
Hence, in our tests over 32-bit characters, NH generates 64-bit hash values whereas the Multilinear families generate 32-bit hash values, but both have a collision probability bounded by $1/2^{32}$. 
Thus,  while NH saves memory because it uses nearly half the number of random bits compared to our fast Multilinear families, Multilinear families may save memory in a system  that stores hash values
because their  hash values have half the number of bits. 
 Table~\ref{table:results3} shows that
  the 64-bit NH on 64-bit processors runs at about the same speed as the best Multilinear on most processors.
Only on two
  Intel Core i7 processors (2600 and 2677M), NH's running time is 60\% of Multilinear's when we enable SSE support. On only one AMD processor (AMD~FX8150), NH is 3 times faster. 
  In other words, sacrificing theoretical guarantees 
does not
always translate into better speed.

\section{Discussion}

Overall, these numbers indicate that strongly universal string hashing
is computationally inexpensive on 
most Intel and AMD processors.
To get good results with older 64-bit and AMD processors, we recommend 
the use of  %intentional verbosity to keep the algo names and cpu 
%names unbroken for this key recommendation
\textsc{Multilinear-HM}. On more recent Intel  processors (i7-2600 and i7-2677M), \textsc{Multilinear}~(2-by-2) is just as fast.

Unfortunately---over long strings---strongly universal hashing requires many random numbers. Generating and storing these random numbers is the main difficulty. 
Whether this is a problem depends on the memory available, 
the CPU cache, the application workload and the length of the strings. 
(Intel researchers reported the
generation of true random numbers
in hardware at high speed (4\,Gbps)~\cite{srinivasan20094gbps}.)
In practice, unexpectedly long strings may require the generation of
new random numbers while hashing a given string~\cite{Crosby:2003:DSV:1251353.1251356}. This overhead should
be relatively inexpensive if we know the length of each string
before we process~it. 

\begin{table}%[tb]   Placement is now ugly, but maybe not as bad as having
%table 4 on page 11 and table 3 on page 12, which is what we had
\begin{minipage}{\columnwidth}\singlespacing
\caption{\label{table:results3}A comparison of estimated CPU cycles per byte between fast Multilinear hashing and the
almost universal hash function NH from Black et al.~\cite{703969} for 32-bit hash values using 64-bit arithmetic. 
When running NH~tests, we remove the \texttt{-mno-sse2} and \texttt{-fno-tree-vectorize} flags, where %it is 
present, to get better results.
}
\centering
\begin{tabular}{c|cc}\hline
                   &  NH~\cite{703969}  & best Multilinear\\ \hline
Intel Core 2 Duo
                   &      0.53 &     \textbf{0.52 }                    \\
Intel Xeon X5260   &         \textbf{0.50}                  &     \textbf{0.50}                  \\
Intel Core i7-860       &     \textbf{0.42}                     &     \textbf{0.42}                     \\
Intel Core i7-2600       &   \textbf{0.16} &     0.27                    \\
Intel Core i7-2677M    &      \textbf{0.12} & 0.20 \\
AMD Sempron 3500+  &         \textbf{0.38} &   0.40                  \\
AMD  V120          &           \textbf{0.38} &0.40                     \\
AMD FX8150         &         \textbf{0.17}                 &      0.51 \\
\hline
\end{tabular}
\end{minipage}
\end{table}

\section{Conclusion}

Over moderately long 32-bit strings ($\approx$1024~characters), current
desktop processors can achieve strongly universal hashing with no more than 0.5~CPU cycle per byte, and 
sometimes as little as 0.2~CPU cycle per byte. Meanwhile, at least twice as many cycles are required for
Rabin-Karp hashing even though it is not even universal. 

While it uses half the number of multiplications, we have found that \textsc{Multilinear-HM} is often no faster than \textsc{Multilinear} on Intel processors.
Clearly, Intel's pipelining architecture has some benefits. 
For %older 
AMD processors,
\textsc{Multilinear-HM} is faster ($\approx$ 33\%), as expected because it uses fewer multiplications. 
%However a more recent AMD processor favors \textsc{Multilinear} over  \textsc{Multilinear-HM}.
 Yet  another alternative,  \textsc{Multilinear} (2-by-2), was slightly faster ($\approx$  15\%) for 32-bit hashing 
on the mobile ARM-based
processors even though it requires twice as many multiplications as \textsc{Multilinear-HM}. These mobile ARM-based processors also computed  
32-bit Rabin-Karp hashing with fewer cycles per byte 
than many desktop processors. We believe that this is related to the
presence of a multiply-accumulate in the ARM instruction set.

%\section*{Acknowledgments}
\ack

The authors are grateful for related online discussions with N.~Kurz,
P.~L.~Bannister, V.~Khare, V.~Venugopal, and J.~S.~Culpepper. 
% Computer Journal wants separate section
%This work is supported by NSERC grant~261437. 
P.~Crowley contributed a correction regarding the extension of \textsc{Multilinear-HM} to variable-length strings.

%Computer Journal wants a separate section
\section*{Funding}
This work was supported by  the Natural Sciences and Engineering Research Council of Canada [261437].

%%%%%%% computer journal
\bibliographystyle{compj}
\bibliography{truncated} 

\appendix

%if we cannot use minipages, we will have to remember to fiddle with this
%after every revision because we don't seem to be protected from stupid
%line and column breaking issues

\section{Implementations in C}  %Code in C might go w App B better
\label{appendix:implementations}
We implemented the following hash functions:
\begin{itemize} 
\item \textsc{Multilinear}:\\$h(s)=m_{1} + \sum_{i=1}^{n} m_{i+1} s_{i}$\\
\item \textsc{Multilinear}~(2-by-2):\\ $h(s)=m_{1} + \sum_{i=1}^{n/2} m_{2i} s_{2i-1}+ s_{2i} m_{2i+1}$\\ 
\item \textsc{Multilinear-HM}:\\ $h(s)=m_{1} + \sum_{i=1}^{n/2} (m_{2i} + s_{2i-1})(s_{2i} + m_{2i+1})$ 
\end{itemize}
For simplicity we assume that the number of characters ($n$) is even.
Following a common convention, we write the unsigned 32-bit and 64-bit integer data types as \texttt{uint32} and \texttt{uint64}.
The variable \texttt{p} is a pointer to the initial value of the string whereas \texttt{endp}
is a pointer to the location right after the last 32-bit character of the \emph{string}.
The variable \texttt{m} is a pointer to the 64-bit random numbers. (When using
63-bit random numbers as allowed by Theorem~\ref{thm:multilinearisstrong}, the right shifts should be by 31 instead. In practice, we use 64-bit numbers.) On some compilers and processors, it was useful to disable
SSE2: under GNU~GCC we can achieve this result with function attributes (e.g.\ by preceding the function declaration by \verb+__attribute__ ((__target__ ("no-sse2")))+).

\lstset{language=C, showstringspaces=false, breaklines,
 emph={p}, emphstyle={\textbf},
 emph={[2]endp}, emphstyle={[2]\textbf},
emph={[3]m}, emphstyle={[3]\textbf},
emph={[4]sum}, emphstyle={[4]\textbf},
emph={[5]uint32}, emphstyle={[5]\textbf},
emph={[6]uint64}, emphstyle={[6]\textbf},
emph={[7]__m128i}, emphstyle={[7]\textbf}
}

{\small\singlespacing\paragraph{\textsc{Multilinear}}~
\begin{lstlisting}
uint32 hash(uint64 * m, uint32 * p, 
    uint32 * endp) {
  uint64 sum = *(m++);
  for(;p != endp; ++m, ++p)
    sum+= *m * *p;
  return sum >> 32;
}
\end{lstlisting}}
 
{\singlespacing\small
\paragraph{\textsc{Multilinear}~(2-by-2)}~
\begin{lstlisting}
uint32 hash(uint64 * m, uint32 * p, 
      uint32 * endp) {
    uint64 sum = *(m++);
    for(; p != endp; m += 2, p += 2)
      sum += (*m * *p) + (*(m+1) * *(p+1));
    return sum >> 32;
}
\end{lstlisting}}

{\singlespacing\small
\paragraph{\textsc{Multilinear-HM}}~
\begin{lstlisting}      
uint32 hash(uint64 * m, uint32 * p, 
    uint32 * endp) {
  uint64 sum = *(m++);
  for(; p != endp; m += 2, p += 2) {
    sum += (*m + *p) * (*(m+1) + *(p+1));
  }
  return sum >> 32;
}
\end{lstlisting}}

\section{Code with CLMUL}
\label{appendix:implementationscarryless}
We implemented Multilinear in $GF(2^{32})$ in C using
the Carry-less Multiplication (CLMUL) instruction set~\citep{Gueron2010549} supported by recent  Intel and AMD processors. We also implemented the counterpart to
\textsc{Multilinear-HM} which executes half the number of multiplications.

We use the same conventions as in~\ref{appendix:implementations} regarding the
variables \texttt{p} and \texttt{m} except that the latter is a pointer to 32-bit random numbers. We wrote our
C programs using SSE~intrinsics: they are
functions supported by 
several major compilers (including GNU~GCC, Intel and Microsoft) that generate SIMD instructions.

The Barrett reduction algorithm is adapted from Kne\v{z}evi{\'c} et al.~\cite{springerlink:10.1007/978-3-540-69499-1_7}.
The variable \texttt{C} contains the chosen irreducible polynomial. We initialize it as
\begin{eqnarray*}
\texttt{C} & = & \texttt{\_mm\_set\_epi64x(0,1UL} \texttt{+~(1UL<<2)} \texttt{+~(1UL<<6)} \\ &&\texttt{+~(1UL<<7)}\texttt{+~(1UL<<32));}.
\end{eqnarray*}
\vspace*{-5ex}
%\begin{minipage}{\linewidth}
 {\singlespacing\small
 \paragraph{Barrett reduction}~
\begin{lstlisting}   
uint32 barrett( __m128i A) {
  __m128i Q1 = _mm_srli_epi64 (A, n);
  __m128i Q2 = _mm_clmulepi64_si128( Q1, C, 0x00);
  __m128i Q3 = _mm_srli_epi64 (Q2, n);
  __m128i f = _mm_xor_si128 (A, _mm_clmulepi64_si128( Q3, C, 0x00));
  return _mm_cvtsi128_si64(f) ;
}
\end{lstlisting}}
%\end{minipage}

%\begin{minipage}{\linewidth}
 {\singlespacing\small
\paragraph{\textsc{GF Multilinear}}~
%\vspace*{-1ex}
\begin{lstlisting}   
uint32 hash(uint32 * m, uint32 * p, 
  uint32 * endp) {
  __m128i sum = _mm_set_epi64x(0, *(m++));
  for(; p != endp; ++m, ++p ) {
    __m128i t = _mm_set_epi64x(*m, *p);
    __m128i c 
    = _mm_clmulepi64_si128( t, t, 0x10);
    sum = _mm_xor_si128 (c, sum);
  }
  return barret(sum);
}
\end{lstlisting}}

 \vspace*{15ex}

%\end{minipage}
%\vspace*{-2ex}
%\vspace*{5ex}  %stupid attempt to force header onto next column
%\begin{minipage}{\linewidth}
 {\singlespacing\small
\paragraph{\textsc{GF Multilinear-HM}}~
\vspace*{-1ex}
\begin{lstlisting}   
uint32 hash(uint32 * m, uint32 * p, 
    uint32 * endp) {
  __m128i sum = _mm_set_epi64x(0, *(m++));
  for(; p != endp; m += 2, p += 2 ) {
    __m128i t1 = _mm_set_epi64x(*m,*(m+1));
    __m128i t2 = _mm_set_epi64x(*p,*(p+1));
    __m128i t = _mm_xor_si128(t1, t2); 
    __m128i c = _mm_clmulepi64_si128(t, t, 0x10);
    sum = _mm_xor_si128 (c, sum);  
  }
  return barret(sum);
}
\end{lstlisting}
 }
%\end{minipage}

\begin{minipage}{.95\linewidth}
 {\singlespacing\small
\paragraph{\textsc{GF Multilinear-HM-Fast}}~

\begin{lstlisting}   
uint32 hash(uint32 * m, uint32 * p, 
    uint32 * endp) {
 // assume m, p, endp are 128-bit aligned
  __m128i z =  _mm_setzero_si128();
  __m128i sum = _mm_set_epi64x(0, *m);
  m += 4;
  __m128i t, u, t1, t2, ts, c1, c2;
  for(; p != endp; m += 4, p += 4 ) {
    t1 = _mm_load_si128((__m128i *) m);
    t2 = _mm_load_si128((__m128i *) p);
    ts = _mm_xor_si128(t1,t2); 
    t = _mm_unpacklo_epi32(ts,z);
    c1 = _mm_clmulepi64_si128(t, t,0x10);
    sum = _mm_xor_si128(c1,sum);   
    u = _mm_unpackhi_epi32(ts,z);
    c2 = _mm_clmulepi64_si128(u, u,0x10);
    sum = _mm_xor_si128 (c2,sum);     
  }
  return barret(sum);
}
\end{lstlisting}
 }

\end{minipage}
\end{document}